\documentclass[11pt,a4paper]{article}

\usepackage{fullpage}
\usepackage{latexsym}

\usepackage{url}
\usepackage{color}
\usepackage{tikz}
\usetikzlibrary{shapes.symbols}
\usepackage{epsfig}

\usepackage{algorithmic}
\usepackage{algorithm}

\newcommand{\ceiling}[1]{\lceil #1 \rceil}
\newcommand{\floor}[1]{\lfloor #1 \rfloor}
\newcommand{\merge}{\mathsf{stable\_merge}}

\newcommand{\corank}{\mathsf{co\_rank}}
\newcommand{\cand}{\mathbf{and}}

\newtheorem{lemma}{Lemma}
\newtheorem{proposition}{Proposition}
\newenvironment{proof}{\emph{Proof:}}{$\Box$\newline}

\begin{document}

\title{Perfectly load-balanced, optimal, stable, parallel merge}

\author{Christian Siebert \\
RWTH Aachen University, Department of Computer Science\\
Laboratory for Parallel Programming\\
Schinkelstrasse 2a, 52062 Aachen, Germany\\
\texttt{christian.siebert@rwth-aachen.de}
\and
Jesper Larsson Tr\"aff\\
Vienna University of Technology, Institute of Information Systems\\
Research Group Parallel Computing\\
Favoritenstrasse 16, 1040 Wien, Austria\\
\texttt{traff@par.tuwien.ac.at}
}

\maketitle

\begin{abstract}
We present a simple, work-optimal and synchronization-free
solution to the problem of stably merging in parallel two given,
ordered arrays of $m$ and $n$ elements into an ordered array of $m+n$
elements.  The main contribution is a new, simple, fast and direct
algorithm that determines, for any prefix of the stably merged output
sequence, the exact prefixes of each of the two input sequences needed
to produce this output prefix.  More precisely, for any given index
(rank) in the resulting, but not yet constructed output array
representing an output prefix, the algorithm computes the indices
(co-ranks) in each of the two input arrays representing the required
input prefixes \emph{without} having to merge the input arrays.  The
co-ranking algorithm takes $\mathcal{O}(\log\min(m,n))$ time steps.
The algorithm is used to devise a \emph{perfectly load-balanced},
\emph{stable}, parallel merge algorithm where each of $p$ processing
elements has \emph{exactly} the same number of input elements to
merge.  Compared to other approaches to the parallel merge problem,
our algorithm is considerably simpler and can be faster up to a factor
of two. Compared to previous algorithms for solving the co-ranking
problem, the algorithm given here is direct and maintains stability in
the presence of repeated elements at no extra space or time cost. When
the number of processing elements $p$ does not exceed
$(m+n)/\log\min(m,n)$, the parallel merge algorithm has optimal
speedup. It is easy to implement on both shared and distributed memory
parallel systems.
\end{abstract}


\section{Introduction}
\label{sec:intro}

We consider the problem of \emph{stably} merging two ordered sequences
in parallel. We assume the two sequences with $m$ and $n$ elements
respectively to be stored in arrays. Elements have a key, and an
ordering relation denoted by $\leq$ is defined on the keys.  The task
is to produce an output sequence consisting of all input elements in
order. Figure~\ref{fig:merging} illustrates the problem, where the
height of each bar corresponds to the element key. Merging the two
sequences sequentially can be done in $\mathcal{O}(m + n)$ operations
(cf. ~\cite{Knuth73:taocp3}).  Most sequential algorithms are
naturally stable, meaning that the relative order of elements with the
same key is preserved.  Figure~\ref{fig:merging} illustrates this: the
equal-keyed elements $\alpha$ and $\beta$ both occur in the output with the
element $\alpha$ from the $A$ array before the element $\beta$ from
the $B$ array, and with all elements from $A$ before $\alpha$ also occuring
before $\alpha$ in the output array.
Stability is important for many applications.

\begin{figure}
   \begin{center}
%
\begin{tikzpicture}
   \colorlet{colbars}{blue!50!white}	
   \colorlet{colcloud}{yellow!20!white}	
   \colorlet{colAbar}{red!50!white}	
   \colorlet{colBbar}{green!50!white}	

   \node[font=\footnotesize,color=black,anchor=west] at (0.1cm,0.5cm) {$A$};
   \foreach \x/\val in {
	1/0.05,2/0.10,3/0.15,4/0.15,5/0.20,6/0.30,7/0.40,8/0.40,9/0.40,
	10/0.50,11/0.60,12/0.60,13/0.65,14/0.70,15/0.75,16/0.80,17/0.95,
	18/0.95}
   {
      \draw[colbars!50!black,fill=colbars] (0.45cm + \x cm*0.25,0.0cm) --
		(0.7cm + \x cm*0.25,0.0cm) -- (0.7cm + \x cm*0.25,\val cm) --
		(0.45cm + \x cm*0.25,\val cm) -- (0.45cm + \x cm*0.25,0.0cm);
   }
   \draw[colAbar!50!black,fill=colAbar] (0.45cm + 9cm*0.25,0.0cm) --
		(0.7cm + 9cm*0.25,0.0cm) -- (0.7cm + 9cm*0.25,0.40cm) --
		(0.45cm + 9cm*0.25,0.40cm) -- (0.45cm + 9cm*0.25,0.0cm);
   \node[font=\footnotesize,color=black,anchor=south] at (0.575cm + 9cm*0.25,0.40cm) {$\alpha$};

   \draw[black,fill=none] (0.7cm,0.0cm) -- (5.2cm,0.0cm) -- (5.2cm,1.0cm) --
			   (0.7cm,1.0cm) -- (0.7cm,0.0cm);
   \node[font=\footnotesize,color=black,anchor=south] at (0.7cm,1.0cm) {$0$};
   \node[font=\footnotesize,color=black,anchor=south] at (5.2cm,1.0cm) {$m$};

   \node[font=\footnotesize,color=black,anchor=west] at (6.1cm,0.5cm) {$B$};
   \foreach \x/\val in {
	1/0.10,2/0.15,3/0.20,4/0.20,5/0.30,6/0.40,7/0.40,8/0.50,9/0.50,
	10/0.50,11/0.60,12/0.65,13/0.65,14/0.65,15/0.65,16/0.65,17/0.70,
	18/0.75,19/0.75,20/0.80,21/0.85,22/0.85,23/0.90,24/0.95}
   {
      \draw[colbars!50!black,fill=colbars] (6.45cm + \x cm*0.25,0.0cm) --
                (6.7cm + \x cm*0.25,0.0cm) -- (6.7cm + \x cm*0.25,\val cm) --
                (6.45cm + \x cm*0.25,\val cm) -- (6.45cm + \x cm*0.25,0.0cm);
   }
   \draw[colBbar!50!black,fill=colBbar] (6.45cm + 6cm*0.25,0.0cm) --
		(6.7cm + 6cm*0.25,0.0cm) -- (6.7cm + 6cm*0.25,0.40cm) --
		(6.45cm + 6cm*0.25,0.40cm) -- (6.45cm + 6cm*0.25,0.0cm);
   \node[font=\footnotesize,color=black,anchor=south] at (6.575cm + 6cm*0.25,0.40cm) {$\beta$};

   \draw[black,fill=none] (6.7cm,0.0cm) -- (12.7cm,0.0cm) -- (12.7cm,1.0cm) --
			   (6.7cm,1.0cm) -- (6.7cm,0.0cm);
   \node[font=\footnotesize,color=black,anchor=south] at (6.7cm,1.0cm) {$0$};
   \node[font=\footnotesize,color=black,anchor=south] at (12.7cm,1.0cm) {$n$};

   \draw[fill=colcloud] (3.8cm,-0.5cm) -- (4.6cm,-1.5cm) -- (8.8cm,-1.5cm) -- (9.6cm,-0.5cm)
			-- (3.8cm,-0.5cm);
   \node[fill=none] at (6.7cm,-1.0cm) {stable two-way merge};
   \draw[-latex,thick] (3.2cm,0.0cm) .. controls +(0.0cm,-0.7cm) and +(-0.7cm,0.0cm) .. (4.2cm,-1.0cm);
   \draw[-latex,thick] (10.2cm,0.0cm) .. controls +(0.0cm,-0.7cm) and +(0.7cm,0.0cm) .. (9.2cm,-1.0cm);
   \draw[-latex,thick] (6.7cm,-1.5cm) -- (6.7cm,-2.0cm);

   \node[font=\footnotesize,color=black,anchor=west] at (0.85cm,-2.5cm) {$C$};
   \foreach \x/\val in {
	1/0.05,2/0.10,3/0.10,4/0.15,5/0.15,6/0.15,7/0.20,8/0.20,9/0.20,
	10/0.30,11/0.30,12/0.40,13/0.40,14/0.40,15/0.40,16/0.40,17/0.50,
	18/0.50,19/0.50,20/0.50,21/0.60,22/0.60,23/0.60,24/0.65,25/0.65,
	26/0.65,27/0.65,28/0.65,29/0.65,30/0.70,31/0.70,32/0.75,33/0.75,
	34/0.75,35/0.80,36/0.80,37/0.85,38/0.85,39/0.90,40/0.95,41/0.95,42/0.95}
   {
      \draw[colbars!50!black,fill=colbars] (1.2cm + \x cm*0.25,-3.0cm) --
                (1.45cm + \x cm*0.25,-3.0cm) -- (1.45cm + \x cm*0.25,\val cm -3.0cm) --
                (1.2cm + \x cm*0.25,\val cm -3.0cm) -- (1.2cm + \x cm*0.25,-3.0cm);
   }
   \draw[colAbar!50!black,fill=colAbar] (1.2cm + 14cm*0.25,-3.0cm) --
		(1.45cm + 14cm*0.25,-3.0cm) -- (1.45cm + 14cm*0.25,0.40cm -3.0cm) --
		(1.2cm + 14cm*0.25,0.40cm -3.0cm) -- (1.2cm + 14cm*0.25,-3.0cm);
   \node[font=\footnotesize,color=black,anchor=south] at (1.325cm + 14cm*0.25,0.40cm -3.0cm) {$\alpha$};
   \draw[colBbar!50!black,fill=colBbar] (1.2cm + 15cm*0.25,-3.0cm) --
		(1.45cm + 15cm*0.25,-3.0cm) -- (1.45cm + 15cm*0.25,0.40cm -3.0cm) --
		(1.2cm + 15cm*0.25,0.40cm -3.0cm) -- (1.2cm + 15cm*0.25,-3.0cm);
   \node[font=\footnotesize,color=black,anchor=south] at (1.325cm + 15cm*0.25,0.40cm -3.0cm) {$\beta$};

   \draw[black,fill=none] (1.45cm,-3.0cm) -- (11.95cm,-3.0cm) --
		(11.95cm,-2.0cm) -- (1.45cm,-2.0cm) -- (1.45cm,-3.0cm);
   \node[font=\footnotesize,color=black,anchor=south] at (1.45cm,-2.0cm) {$0$};
   \node[font=\footnotesize,color=black,anchor=south] at (11.95cm,-2.0cm) {$m+n$};

\end{tikzpicture}
   \end{center}
   \caption{Merging takes two ordered input sequences stored in arrays
     $A$ and $B$ and produces an ordered output sequence stored in
     array $C$.  Stability preserves the original order of equal-keyed
     elements, such as $\alpha$ and $\beta$, with elements from $A$
     occuring before elements from $B$.}
   \label{fig:merging}
\end{figure}
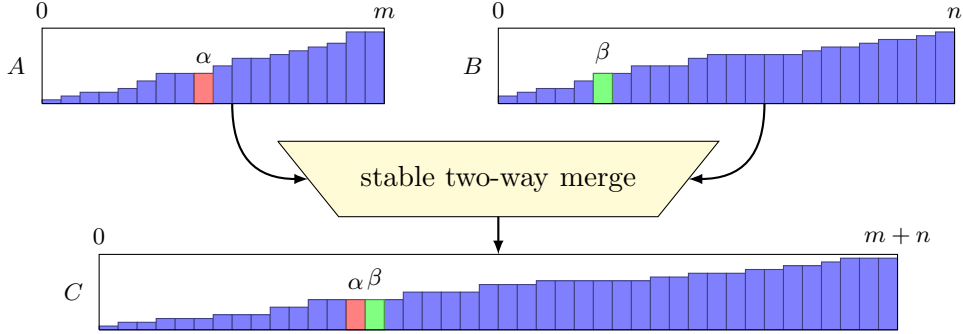

Many parallel two-way merge algorithms for large, ordered input
arrays work as follows. In each of the arrays, a number proportional
to the number of available processing elements of fixed, equidistant
elements are chosen. For each the \emph{(cross)rank}, which is the
number of elements smaller than the chosen element in the other array,
is determined via binary search. Cross ranks and indices of the chosen
elements are used to determine disjoint array segments of the two
input arrays that can be merged independently in parallel. On a
parallel machine with $p$ processing elements, the array segments to
be merged by a single processing element can be guaranteed to be of
size at most $\ceiling{m/p}$ and $\ceiling{n/p}$ respectively, but the size
of the segments for different processing elements may differ by a
factor of two. This pattern is found for instance in well-known PRAM
(Parallel Random Access Machine, see, e.g.,~\cite{JaJa92}) and BSP
(Bulk Synchronous Parallel model, see, e.g.,~\cite{Valiant90:bridge})
algorithms~\cite{GerbessiotisSiniolakis01,HagerupRub89,KatajainenLevcopoulosPetersson93,ShiloachVishkin81}. Many
of these algorithms use a separate merge step to determine the pairs of
segments to merge. However, this extra step can be eliminated as shown
in~\cite{Traff12:merge}. All of these parallel algorithms take
$\mathcal{O}((n+m)/p+\max(\log n,\log m))$ operations, are therefore
work-optimal for $p \leq (m+n)/\max(\log n,\log m)$, such that the speedup
compared to an optimal sequential algorithm is
$\mathcal{O}(p)$. However, the load imbalance caused by the inexact
determination of segments can limit the speedup to
$p/2$. It is important to note that the dominant part of work in
both parallel and sequential merge algorithm is done by the
\emph{same} (best) sequential merge algorithm.

In this paper, we show how the segments to merge can be determined
such that all processing elements get \emph{exactly} the same number of
elements to merge, with a surprisingly simple and intuitive idea. For
any index (\emph{rank}) representing a prefix of the not yet computed,
ordered output sequence the approach exactly determines the prefix
(\emph{co-rank}) of each of the input sequences that is needed to make
up the given output prefix when the inputs are stably merged. To merge
in parallel, the processing elements are simply assigned disjoint
segments of the output (of roughly equal size) and they use the co-ranking
algorithm to determine disjoint blocks of the two input arrays to
merge. No synchronization is needed between the processing elements.
For any given output rank $i,0\leq i<m+n$, co-ranks can be determined in
$\mathcal{O}(\log\min(m,n))$ operations and parallel merging can therefore
be carried out in $\mathcal{O}((m+n)/p +\log(\min(m,n)))$ operations
per processing element.

This approach closes the gap up to a lower order term of
$\mathcal{O}(\log(\min(m,n)))$ between the actual number of operations
to be carried out by parallel and sequential algorithms. Parallel
algorithms that do the same number of operations with the same
constant factors as the corresponding, best sequential algorithms are
rare. A balanced number of operations is important, because any
additional constant factor overhead leads to a proportional loss of
processing resources. In addition to these performance differences,
the varying number of elements per processing element influences the
memory consumption per processing element, which can be problematic
especially for distributed-memory architectures.

Guaranteeing stability is sometimes problematic for parallel merge
algorithms.  Whenever stability is required by the application, a
standard trick is to merge according to a lexicographic order on
key-index pairs.  From a practical point of view, this technique leads
to undesirable extra compute and space costs. In the algorithm
presented here, stability comes at no extra costs---neither
additional index comparisons nor space consuming
lexicographic orderings are required.

The co-ranking idea is not totally new. It was introduced
in~\cite{AklSantoro87} and used subsequently
in~\cite{DeoJainMedidi94,DeoSarkar91}
and~\cite{VarmanIyerHaderleDunn90,VarmanScheuflerIyerRicard91}, but
seems to have been somewhat
overseen. In~\cite{DeoJainMedidi94,DeoSarkar91} the co-ranking problem
is solved only for the median, and the general problem reduced to the
median case. The specific contribution in this paper is a simple,
direct implementation that works for any output index and is stable by
design. A distributed-memory implementation of a previous version is
described in~\cite{Traff12:mpimerge}. Some recent algorithms for
merging achieve similarly good partitions~\cite{GreenMcCollBader12,OdehGreenMwassiShuuliBirk12} and were in fact inspired
directly by the algorithms in~\cite{DeoJainMedidi94,DeoSarkar91}. We
think the co-ranking algorithm presented next is more intuitive,
with a simpler proof, and slightly better bounds.

\section{A co-ranking algorithm}

Let $A$ and $B$ be the two input arrays with $m$ and $n$ elements,
respectively. We follow C programming language conventions and index
arrays from $0$.  Both arrays are ordered according to an ordering
relation $\leq$ denoting comparison of element keys such that $A[j-1]
\leq A[j]$ for $1\leq j< m$, and $B[k-1] \leq B[k]$ for $1\leq
k<n$. Ultimately, we are interested in performing a stable merge of
$A$ and $B$ into an array $C$ with $m+n$ elements.  We denote this by
$C=\merge(A,B,\leq)$. For any $i, 0\leq i<m+n$ in $C$ there is either a
$j, 0\leq j<m$ such that $C[i]=A[j]$ or a $k, 0\leq k<n$
such that $C[i]=B[k]$.  Furthermore, for any $i$-element prefix
$C[0,\ldots, i-1]$ of $C$ there must be indices $j$ and $k$ of $A$ and
$B$ such that
$C[0,\ldots,i-1]=\merge(A[0,\ldots,j-1],B[0,\ldots,k-1],\leq)$. We
will show in Lemma~\ref{lem:coranks} that these $j$ and $k$ indices are
indeed \emph{unique}. They define the prefixes of $A$ and $B$ needed
to form the prefix of $C$ of length $i$. For an
element $C[i]$ we call the index $i$ its \emph{rank}, and the
unique indices $j$ and $k$ its \emph{co-ranks}.
Consequently, we use the term \emph{co-ranking}
for the process of determining $j$ and $k$ from $A$, $m$, $B$, $n$ and $i$.
Figure~\ref{fig:coranking} illustrates the co-rank definition and process.

\begin{figure}
   \begin{center}
%
\begin{tikzpicture}
   \colorlet{colbars}{blue!50!white}	
   \colorlet{colbarsoff}{gray!50}	
   \colorlet{colcloud}{yellow!20!white}	
   \colorlet{colranks}{red!80!black}	
   \colorlet{colsmall}{blue!50!white}	
   \colorlet{collarge}{green!50!white}	
   \colorlet{colsmallC}{colsmall!30!white}	
   \colorlet{collargeC}{collarge!30!white}	

   \node[font=\footnotesize,color=black,anchor=west] at (0.1cm,0.5cm) {$A$};
   \foreach \x/\val in {
	1/0.05,2/0.10,3/0.15,4/0.15,5/0.20,6/0.30,7/0.40,8/0.40,9/0.40,
	10/0.50,11/0.60,12/0.60}
   {
      \draw[colsmall!50!black,fill=colsmall] (0.45cm + \x cm*0.25,0.0cm) --
		(0.7cm + \x cm*0.25,0.0cm) -- (0.7cm + \x cm*0.25,\val cm) --
		(0.45cm + \x cm*0.25,\val cm) -- (0.45cm + \x cm*0.25,0.0cm);
   }
   \foreach \x/\val in {13/0.65,14/0.70,15/0.75,16/0.80,17/0.95,18/0.95}
   {
      \draw[collarge!50!black,fill=collarge] (0.45cm + \x cm*0.25,0.0cm) --
		(0.7cm + \x cm*0.25,0.0cm) -- (0.7cm + \x cm*0.25,\val cm) --
		(0.45cm + \x cm*0.25,\val cm) -- (0.45cm + \x cm*0.25,0.0cm);
   }
   \draw[black,fill=none] (0.7cm,0.0cm) -- (5.2cm,0.0cm) -- (5.2cm,1.0cm) --
			   (0.7cm,1.0cm) -- (0.7cm,0.0cm);
   \draw[colranks,thick] (0.7cm + 12cm*0.25,1.2cm) -- (0.7cm + 12cm*0.25,0.0cm);
   \node[font=\footnotesize,color=black,anchor=south] at (0.7cm + 12cm*0.25,1.2cm) {$j$};

   \node[font=\footnotesize,color=black,anchor=west] at (6.1cm,0.5cm) {$B$};
   \foreach \x/\val in {1/0.10,2/0.15,3/0.20,4/0.20,5/0.30,6/0.40,7/0.40,8/0.50,9/0.50,10/0.50}
   {
      \draw[colsmall!50!black,fill=colsmall] (6.45cm + \x cm*0.25,0.0cm) --
                (6.7cm + \x cm*0.25,0.0cm) -- (6.7cm + \x cm*0.25,\val cm) --
                (6.45cm + \x cm*0.25,\val cm) -- (6.45cm + \x cm*0.25,0.0cm);
   }
   \foreach \x/\val in {
	11/0.60,12/0.65,13/0.65,14/0.65,15/0.65,16/0.65,17/0.70,
	18/0.75,19/0.75,20/0.80,21/0.85,22/0.85,23/0.90,24/0.95}
   {
      \draw[collarge!50!black,fill=collarge] (6.45cm + \x cm*0.25,0.0cm) --
                (6.7cm + \x cm*0.25,0.0cm) -- (6.7cm + \x cm*0.25,\val cm) --
                (6.45cm + \x cm*0.25,\val cm) -- (6.45cm + \x cm*0.25,0.0cm);
   }
   \draw[black,fill=none] (6.7cm,0.0cm) -- (12.7cm,0.0cm) -- (12.7cm,1.0cm) --
			   (6.7cm,1.0cm) -- (6.7cm,0.0cm);
   \draw[colranks,thick] (6.7cm + 10cm*0.25,1.2cm) -- (6.7cm + 10cm*0.25,0.0cm);
   \node[font=\footnotesize,color=black,anchor=south] at (6.7cm + 10cm*0.25,1.2cm) {$k$};

   \draw[fill=colcloud] (5.6cm,-0.5cm) -- (4.8cm,-1.5cm) -- (8.6cm,-1.5cm) -- (7.8cm,-0.5cm)
			-- (5.6cm,-0.5cm);
   \node[fill=none] at (6.7cm,-1.0cm) {co-ranking};
   \draw[latex-,thick] (0.7cm + 12cm*0.25,0.0cm) .. controls +(0.0cm,-0.7cm) and +(-0.7cm,0.0cm) .. (5.2cm,-1.0cm);
   \draw[latex-,thick] (6.7cm + 10cm*0.25,0.0cm) .. controls +(0.0cm,-0.7cm) and +(0.7cm,0.0cm) .. (8.2cm,-1.0cm);
   \draw[latex-,thick] (6.7cm,-1.5cm) -- (1.45cm + 22cm*0.25,-2.0cm);

   \node[font=\footnotesize,color=black,anchor=west] at (0.85cm,-2.5cm) {$C$};
   \foreach \x/\val in {
	1/0.05,2/0.10,3/0.10,4/0.15,5/0.15,6/0.15,7/0.20,8/0.20,9/0.20,
	10/0.30,11/0.30,12/0.40,13/0.40,14/0.40,15/0.40,16/0.40,17/0.50,
	18/0.50,19/0.50,20/0.50,21/0.60,22/0.60}
   {
      \draw[colsmallC!50!black,fill=colsmallC] (1.2cm + \x cm*0.25,-3.0cm) --
                (1.45cm + \x cm*0.25,-3.0cm) -- (1.45cm + \x cm*0.25,\val cm -3.0cm) --
                (1.2cm + \x cm*0.25,\val cm -3.0cm) -- (1.2cm + \x cm*0.25,-3.0cm);
   }
   \foreach \x/\val in {
	23/0.60,24/0.65,25/0.65,26/0.65,27/0.65,28/0.65,29/0.65,
	30/0.70,31/0.70,32/0.75,33/0.75,34/0.75,35/0.80,36/0.80,
	37/0.85,38/0.85,39/0.90,40/0.95,41/0.95,42/0.95}
   {
      \draw[collargeC!50!black,fill=collargeC] (1.2cm + \x cm*0.25,-3.0cm) --
                (1.45cm + \x cm*0.25,-3.0cm) -- (1.45cm + \x cm*0.25,\val cm -3.0cm) --
                (1.2cm + \x cm*0.25,\val cm -3.0cm) -- (1.2cm + \x cm*0.25,-3.0cm);
   }
   \draw[black,fill=none] (1.45cm,-3.0cm) -- (11.95cm,-3.0cm) --
		(11.95cm,-2.0cm) -- (1.45cm,-2.0cm) -- (1.45cm,-3.0cm);
   \draw[colranks,thick] (1.45cm + 22cm*0.25,-2.0cm) -- (1.45cm + 22cm*0.25,-3.2cm);
   \node[font=\footnotesize,color=black,anchor=north] at (1.45cm + 22cm*0.25,-3.2cm) {$i$};

\end{tikzpicture}

   \end{center}
   \caption{Co-ranking determines for any given rank (index) $i$ in
     $C$ the co-ranks $j$ and $k$ in $A$ and $B$ without having
     to actually merge $A$ and $B$.}
   \label{fig:coranking}
\end{figure}
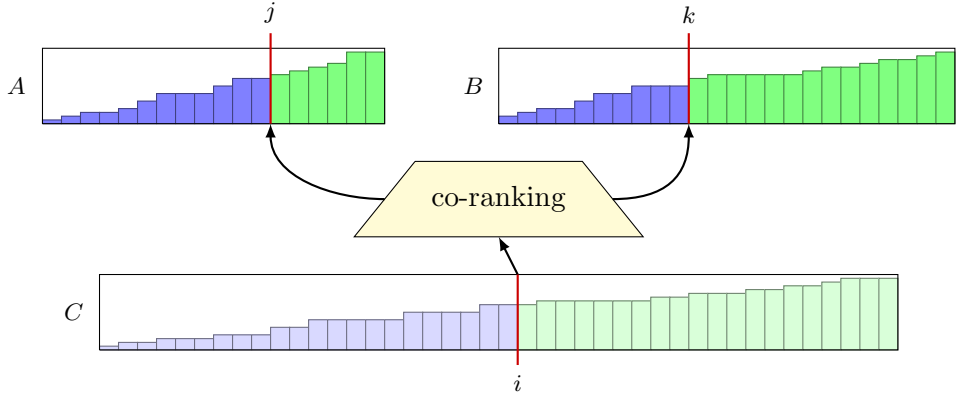

Stability means that all equal elements of $A$ should (in their
relative order in $A$) come before equal elements of $B$ (also in
their relative order in $B$). That is, if $C[i]=A[j]$ and
$A[j]=A[j+1]$ then also $C[i+1]=A[j+1]$, and if $A[j]=B[k]$ then there
is some $i'>i$ for which $C[i']=B[k]$.  Stability is mostly easy to
guarantee for sequential merging.  We can therefore assume that we
have an optimal sequential algorithm for stable merging at our
disposal.

Let $C[i-1]$ be the $i^{\rm th}$ element in the stably merged output array
$C$. For determining the co-ranks $j$ and $k$ (and thereby determining
whether $C[i-1]$ comes from the $A$ or the $B$ array), we first note
that $j+k=i$. The $i^{\rm th}$ output element $C[i-1]$ is either $A[j-1]$ or
$B[k-1]$, both elements $A[j-1]$ and $B[k-1]$ are in the prefix
$C[0,\ldots,i-1]$, but neither of $A[j]$ nor $B[k]$ are. (For
convenience we assume that $A[-1]=-\infty, A[m]=\infty$, and
likewise for $B$. However, these sentinels do not have to be stored.) If
$C[i-1]=A[j-1]$ (that is, the $i^{\rm th}$ output element comes from $A$),
then it must hold that $A[j-1]\leq B[k]$.  If instead $C[i-1]=B[k-1]$,
then likewise $B[k-1]\leq A[j]$. Now, since the merge is stable, it
cannot be that $B[k-1]=A[j]$ since that would mean that an element of
$B$ equal to an element of $A$ comes before the $A$ element in the
output array. Therefore, in this case $B[k-1]<A[j]$.

\begin{lemma}
\label{lem:coranks}
For any $i, 0\leq i<m+n$, there exists a unique $j$, $0\leq j\leq m$, and
a unique $k$, $0\leq k\leq n$, with $j+k=i$ such that
\begin{enumerate}
\item
\label{lem:cond1}
$j=0\vee A[j-1]\leq B[k]$ and 
\item
\label{lem:cond2}
$k=0\vee B[k-1]<A[j]$. 
\end{enumerate}
These $j$ and $k$ fulfill
$\merge(A[0,\ldots,j-1],B[0,\ldots,k-1],\leq)=C[0,\ldots,i-1]$ where
$C=\merge(A,B,\leq)$.
\end{lemma}

We will refer to (\ref{lem:cond1}) and (\ref{lem:cond2}) as the first
and the second Lemma condition, respectively.

\begin{proof}
We need to argue for the uniqueness of $j$ and $k$. Assume the two
conditions hold for some $0\leq j<m$ with a corresponding $k$ such
that $j+k=i$.  We look at $j+1$ and $j-1$: for $j+1$, $A[j]\leq
B[k-1]$ contradicts $B[k-1]<A[j]$, and for $j-1$, $B[k]<A[j-1]$
contradicts $A[j-1]\leq B[k]$. Thus neither $j-1$ nor $j+1$ can
fulfill both conditions, and since the arrays are ordered, neither can
any other smaller or larger indices. For the existence we argue as
follows. Let $j$ with $0<j<m$ be the highest indexed element such that
$A[j-1]\leq B[k]$ with a corresponding $k, 0<k<n$; since $A[j]>B[k-1]$
the second Lemma condition is also fulfilled. If no such element
exists, we check for $j=0$ and $k=0$. For $j=0$ the first lemma
condition is trivially true, and the second lemma condition
$B[k-1]<A[0]$ holds if the length $k$ prefix of $B$ comes
before $A$ in the output array. For $k=0$ the second lemma condition
is trivially true, and the first lemma condition $A[j-1]\leq B[0]$
holds if the length $j$ prefix of $A$ comes before $B$ in the output array.
Since either must be true, indices $j$ and $k$ with $j+k=i$ must exist
as claimed.
\end{proof}

\begin{algorithm}
\caption{$\corank(i,A,m,B,n,\leq)$ determines for any given $i, 0\leq i<m+n$
	 the unique co-ranks $j$ and $k$ in arrays $A$ and $B$.}
\label{alg:corank}
\begin{algorithmic}[1]
\STATE $j\gets \min(i,m)$ 
\STATE $k\gets i-j$ \COMMENT{Invariant: $j+k=i$}
\STATE $j_{\rm low} \gets \max(0,i-n)$
\STATE $\mbox{active} \gets \TRUE$
\WHILE {active}
\IF{$j>0\wedge k<n\ \cand\ A[j-1]>B[k]$ }
\STATE \COMMENT{First Lemma condition violated: decrease $j$}
\STATE $\delta \gets \ceiling{\frac{j-j_{\rm low}}{2}}$
\STATE $k_{\rm low} \gets k$
\STATE $j,k\gets j-\delta,k+\delta$
\ELSIF{$k>0\wedge j<m\ \cand\ B[k-1]\geq A[j]$ }
\STATE \COMMENT{Second Lemma condition violated: decrease $k$}
\STATE $\delta \gets \ceiling{\frac{k-k_{\rm low}}{2}}$
\STATE $j_{\rm low} \gets j$
\STATE $j,k\gets j+\delta,k-\delta$
\ELSE
\STATE \COMMENT{No conditions violated: unique $(j,k)$ found}
\STATE $\mbox{active} \gets \FALSE$
\ENDIF
\ENDWHILE
\RETURN $(j,k)$
\end{algorithmic}
\end{algorithm}

Lemma~\ref{lem:coranks} immediately gives an approach to find the
co-ranks $j$ and $k$ for any $i$ efficiently.
Algorithm~\ref{alg:corank} implements this approach and finds the
unique $j$ and $k$ fulfilling both conditions of
Lemma~\ref{lem:coranks}.  It maintains the invariant $j+k=i$, and
works similar to a binary search.  The $\cand$ in lines 6 and 11
denotes a ``conditional and'' as in the C language and means that the
condition at the right is only evaluated if the condition on the left
evaluates to true. The algorithm starts with the extreme assumption
that all $i$ elements come from $A$ (as far as possible), that is
setting $j=\min(i,m)$ and $k=i-j$.  For both arrays $A$ and $B$ it
maintains lower bounds $j_{\rm low}$, $j_{\rm low} \leq j$, and
$k_{\rm low}$, $k_{\rm low} \leq k$ on the prefix lengths.
Furthermore, it maintains the invariant that either the unique index
in array $A$ fulfilling the conditions of Lemma~\ref{lem:coranks} is
between $j_{\rm low}$ and $j$ \emph{or} the unique index in array $B$
fulfilling the condition is between $k_{\rm low}$ and $k$.

If the first Lemma condition is violated, that is $A[j-1]>B[k]$, then
$j$ is too large and therefore must be decreased. This is done by
cutting the size of the interval from $j_{\rm low}$ to $j$ in half. In
order to maintain the invariant $j+k=i$, $k$ is correspondingly
increased at the same time. In order to prevent that $k$ exceeds $n$,
the lower bound $j_{\rm low}$ is initially set to $\max(0,i-n)$.  When
$j$ is decreased, the lower bound $k_{\rm low}$ can be increased to
$k$ because no smaller $k$ could fulfill the first Lemma condition
(that would mean a larger $j$, which cannot be). If instead the second
Lemma condition is violated, that is $B[k-1]\geq A[j]$, then $k$ needs
to be decreased. Again, we do this by halving the size of the interval
from $k_{\rm low}$ to $k$. Since no smaller $j$ can now fulfill the
first Lemma condition, the lower index $j_{\rm low}$ can be increased
to $j$. The algorithm terminates when both conditions are fulfilled,
which will happen at the latest when either $j-j_{\rm low}=0$ or
$k-k_{\rm low}=0$.  Note that in the first iteration, only the first
Lemma condition can be violated, as per initialization either $j=m$ or
$k=0$. The lower bound $k_{\rm low}$ will therefore be set in the
first iteration.  A first few possible iterations of the algorithm are
illustrated in Figure~\ref{fig:searching}.

\begin{figure}
\begin{center}
\begin{tikzpicture}
   \colorlet{colbars}{blue!50!white}	
   \colorlet{colbarsoff}{gray!50}	
   \colorlet{colranks}{red!80!black}	

   \node[font=\footnotesize,color=black,anchor=west] at (0.0cm,1.4cm)
			{Before:};

   \node[font=\footnotesize,color=black,anchor=west] at (0.1cm,0.5cm) {$A$};
   \foreach \x/\val in {
	1/0.05,2/0.10,3/0.15,4/0.15,5/0.20,6/0.30,7/0.40,8/0.40,9/0.40,
	10/0.50,11/0.60,12/0.60,13/0.65,14/0.70,15/0.75,16/0.80,17/0.95,
	18/0.95}
   {
      \draw[colbars!50!black,fill=colbars] (0.45cm + \x cm*0.25,0.0cm) --
		(0.7cm + \x cm*0.25,0.0cm) -- (0.7cm + \x cm*0.25,\val cm) --
		(0.45cm + \x cm*0.25,\val cm) -- (0.45cm + \x cm*0.25,0.0cm);
   }
   \draw[black,fill=none] (0.7cm,0.0cm) -- (5.2cm,0.0cm) -- (5.2cm,1.0cm) --
			   (0.7cm,1.0cm) -- (0.7cm,0.0cm);
   \node[font=\footnotesize,color=black,anchor=south] at (5.2cm,1.0cm) {$m$};
   \draw[colranks,thick] (0.7cm + 0cm*0.25,1.0cm) -- (0.7cm + 0cm*0.25,-0.2cm);
   \node[font=\footnotesize,color=black,anchor=north] at (0.7cm + 0cm*0.25,-0.2cm) {$j_{\rm low}$};
   \draw[colranks,thick] (0.7cm + 18cm*0.25,1.0cm) -- (0.7cm + 18cm*0.25,-0.2cm);
   \node[font=\footnotesize,color=black,anchor=north] at (0.7cm + 18cm*0.25,-0.2cm) {$j$};

   \node[font=\footnotesize,color=black,anchor=west] at (6.1cm,0.5cm) {$B$};
   \foreach \x/\val in {
	1/0.10,2/0.15,3/0.20,4/0.20,5/0.30,6/0.40,7/0.40,8/0.50,9/0.50,
	10/0.50,11/0.60,12/0.65,13/0.65,14/0.65,15/0.65,16/0.65,17/0.70,
	18/0.75,19/0.75,20/0.80,21/0.85,22/0.85,23/0.90,24/0.95}
   {
      \draw[colbars!50!black,fill=colbars] (6.45cm + \x cm*0.25,0.0cm) --
                (6.7cm + \x cm*0.25,0.0cm) -- (6.7cm + \x cm*0.25,\val cm) --
                (6.45cm + \x cm*0.25,\val cm) -- (6.45cm + \x cm*0.25,0.0cm);
   }
   \draw[black,fill=none] (6.7cm,0.0cm) -- (12.7cm,0.0cm) -- (12.7cm,1.0cm) --
			   (6.7cm,1.0cm) -- (6.7cm,0.0cm);
   \node[font=\footnotesize,color=black,anchor=south] at (12.7cm,1.0cm) {$n$};
   \draw[colranks,thick] (6.7cm + 4cm*0.25,1.0cm) -- (6.7cm + 4cm*0.25,-0.2cm);
   \node[font=\footnotesize,color=black,anchor=north] at (6.7cm + 4cm*0.25,-0.2cm) {$k=i-m$};

   \node[font=\footnotesize,color=black,anchor=west] at (0.0cm,-1.1cm)
			{Iteration 1 (First Lemma condition violated):};

   \node[font=\footnotesize,color=black,anchor=west] at (0.1cm,-2.0cm) {$A$};
   \foreach \x/\val in {
	1/0.05,2/0.10,3/0.15,4/0.15,5/0.20,6/0.30,7/0.40,8/0.40,9/0.40}
   {
      \draw[colbars!50!black,fill=colbars] (0.45cm + \x cm*0.25,-2.5cm) --
		(0.7cm + \x cm*0.25,-2.5cm) -- (0.7cm + \x cm*0.25,\val cm -2.5cm) --
		(0.45cm + \x cm*0.25,\val cm -2.5cm) -- (0.45cm + \x cm*0.25,-2.5cm);
   }
   \foreach \x/\val in {
	10/0.50,11/0.60,12/0.60,13/0.65,14/0.70,15/0.75,16/0.80,17/0.95,18/0.95}
   {
      \draw[colbarsoff!50!black,fill=colbarsoff] (0.45cm + \x cm*0.25,-2.5cm) --
		(0.7cm + \x cm*0.25,-2.5cm) -- (0.7cm + \x cm*0.25,\val cm -2.5cm) --
		(0.45cm + \x cm*0.25,\val cm -2.5cm) -- (0.45cm + \x cm*0.25,-2.5cm);
   }
   \draw[black,fill=none] (0.7cm,-2.5cm) -- (5.2cm,-2.5cm) -- (5.2cm,-1.5cm) --
			   (0.7cm,-1.5cm) -- (0.7cm,-2.5cm);
   \draw[colranks,thick] (0.7cm + 0cm*0.25,-1.5cm) -- (0.7cm + 0cm*0.25,-2.7cm);
   \node[font=\footnotesize,color=black,anchor=north] at (0.7cm + 0cm*0.25,-2.7cm) {$j_{\rm low}$};
   \draw[colranks,thick] (0.7cm + 9cm*0.25,-1.5cm) -- (0.7cm + 9cm*0.25,-2.7cm);
   \node[font=\footnotesize,color=black,anchor=north] at (0.7cm + 9cm*0.25,-2.7cm) (it1j) {$j$};
   \node[font=\footnotesize,color=black] at (0.7cm + 13.5cm*0.25,-5cm |- it1j) 
	(it1delta) {$\delta$};
   \draw[black,latex-] (it1j.east) -- (it1delta.west);
   \draw[black] (it1delta.east) -- (0.7cm + 18cm*0.25,-5cm |- it1j);

   \node[font=\footnotesize,color=black,anchor=west] at (6.1cm,-2.0cm) {$B$};
   \foreach \x/\val in {
	5/0.30,6/0.40,7/0.40,8/0.50,9/0.50,   10/0.50,11/0.60,12/0.65,13/0.65}
   {
      \draw[colbars!50!black,fill=colbars] (6.45cm + \x cm*0.25,-2.5cm) --
                (6.7cm + \x cm*0.25,-2.5cm) -- (6.7cm + \x cm*0.25,\val cm - 2.5cm) --
                (6.45cm + \x cm*0.25,\val cm - 2.5cm) -- (6.45cm + \x cm*0.25,-2.5cm);
   }
   \foreach \x/\val in {
	1/0.10,2/0.15,3/0.20,4/0.20,   14/0.65,15/0.65,16/0.65,17/0.70,
	18/0.75,19/0.75,20/0.80,21/0.85,22/0.85,23/0.90,24/0.95}
   {
      \draw[colbarsoff!50!black,fill=colbarsoff] (6.45cm + \x cm*0.25,-2.5cm) --
                (6.7cm + \x cm*0.25,-2.5cm) -- (6.7cm + \x cm*0.25,\val cm - 2.5cm) --
                (6.45cm + \x cm*0.25,\val cm - 2.5cm) -- (6.45cm + \x cm*0.25,-2.5cm);
   }
    \draw[black,fill=none] (6.7cm,-2.5cm) -- (12.7cm,-2.5cm) -- (12.7cm,-1.5cm) --
			   (6.7cm,-1.5cm) -- (6.7cm,-2.5cm);
   \draw[colranks,thick] (6.7cm + 4cm*0.25,-1.5cm) -- (6.7cm + 4cm*0.25,-2.7cm);
   \node[font=\footnotesize,color=black,anchor=north] at (6.7cm + 4cm*0.25,-2.7cm)
	(it1klow) {$k_{\rm low}$};
   \draw[colranks,thick] (6.7cm + 13cm*0.25,-1.5cm) -- (6.7cm + 13cm*0.25,-2.7cm);
   \node[font=\footnotesize,color=black,anchor=north] at (6.7cm + 13cm*0.25,-2.7cm) (it1k) {$k$};
   \node[font=\footnotesize,color=black] at (6.7cm + 8.5cm*0.25,-5cm |- it1j) 
	(it1delta2) {$\delta$};
   \draw[black] (it1klow.east |- it1j) -- (it1delta2.west);
   \draw[black,-latex] (it1delta2.east) -- (it1k.west |- it1j);

   \node[font=\footnotesize,color=black,anchor=west] at (0.0cm,-3.6cm)
			{Iteration 2 (Second Lemma condition violated):};

   \node[font=\footnotesize,color=black,anchor=west] at (0.1cm,-4.5cm) {$A$};
   \foreach \x/\val in {
	10/0.50,11/0.60,12/0.60,13/0.65,14/0.70}
   {
      \draw[colbars!50!black,fill=colbars] (0.45cm + \x cm*0.25,-5.0cm) --
		(0.7cm + \x cm*0.25,-5.0cm) -- (0.7cm + \x cm*0.25,\val cm -5.0cm) --
		(0.45cm + \x cm*0.25,\val cm -5.0cm) -- (0.45cm + \x cm*0.25,-5.0cm);
   }
   \foreach \x/\val in {
	1/0.05,2/0.10,3/0.15,4/0.15,5/0.20,6/0.30,7/0.40,8/0.40,9/0.40,
	15/0.75,16/0.80,17/0.95,18/0.95}
   {
      \draw[colbarsoff!50!black,fill=colbarsoff] (0.45cm + \x cm*0.25,-5.0cm) --
		(0.7cm + \x cm*0.25,-5.0cm) -- (0.7cm + \x cm*0.25,\val cm -5.0cm) --
		(0.45cm + \x cm*0.25,\val cm -5.0cm) -- (0.45cm + \x cm*0.25,-5.0cm);
   }
   \draw[black,fill=none] (0.7cm,-5.0cm) -- (5.2cm,-5.0cm) -- (5.2cm,-4.0cm) --
			   (0.7cm,-4.0cm) -- (0.7cm,-5.0cm);
   \draw[colranks,thick] (0.7cm + 9cm*0.25,-4.0cm) -- (0.7cm + 9cm*0.25,-5.2cm);
   \node[font=\footnotesize,color=black,anchor=north] at (0.7cm + 9cm*0.25,-5.2cm)
	(it2jlow) {$j_{\rm low}$};
   \draw[colranks,thick] (0.7cm + 14cm*0.25,-4.0cm) -- (0.7cm + 14cm*0.25,-5.2cm);
   \node[font=\footnotesize,color=black,anchor=north] at (0.7cm + 14cm*0.25,-5.2cm) (it2j) {$j$};
   \node[font=\footnotesize,color=black] at (0.7cm + 11.5cm*0.25,-5cm |- it2j) 
	(it2delta) {$\delta$};
   \draw[black] (it2jlow.east |- it2j) -- (it2delta.west |- it2j);
   \draw[black,-latex] (it2delta.east |- it2j) -- (it2j.west);

   \node[font=\footnotesize,color=black,anchor=west] at (6.1cm,-4.5cm) {$B$};
   \foreach \x/\val in {
	5/0.30,6/0.40,7/0.40,8/0.50}
   {
      \draw[colbars!50!black,fill=colbars] (6.45cm + \x cm*0.25,-5.0cm) --
                (6.7cm + \x cm*0.25,-5.0cm) -- (6.7cm + \x cm*0.25,\val cm - 5.0cm) --
                (6.45cm + \x cm*0.25,\val cm - 5.0cm) -- (6.45cm + \x cm*0.25,-5.0cm);
   }
   \foreach \x/\val in {
	1/0.10,2/0.15,3/0.20,4/0.20,
	9/0.50,10/0.50,11/0.60,12/0.65,13/0.65,14/0.65,15/0.65,16/0.65,
	17/0.70,18/0.75,19/0.75,20/0.80,21/0.85,22/0.85,23/0.90,24/0.95}
   {
      \draw[colbarsoff!50!black,fill=colbarsoff] (6.45cm + \x cm*0.25,-5.0cm) --
                (6.7cm + \x cm*0.25,-5.0cm) -- (6.7cm + \x cm*0.25,\val cm - 5.0cm) --
                (6.45cm + \x cm*0.25,\val cm - 5.0cm) -- (6.45cm + \x cm*0.25,-5.0cm);
   }
    \draw[black,fill=none] (6.7cm,-5.0cm) -- (12.7cm,-5.0cm) -- (12.7cm,-4.0cm) --
			   (6.7cm,-4.0cm) -- (6.7cm,-5.0cm);
   \draw[colranks,thick] (6.7cm + 4cm*0.25,-4.0cm) -- (6.7cm + 4cm*0.25,-5.2cm);
   \node[font=\footnotesize,color=black,anchor=north] at (6.7cm + 4cm*0.25,-5.2cm)
	(it2klow) {$k_{\rm low}$};
   \draw[colranks,thick] (6.7cm + 8cm*0.25,-4.0cm) -- (6.7cm + 8cm*0.25,-5.2cm);
   \node[font=\footnotesize,color=black,anchor=north] at (6.7cm + 8cm*0.25,-5.2cm) (it2k) {$k$};
   \node[font=\footnotesize,color=black] at (6.7cm + 10.5cm*0.25,-5cm |- it2j) 
	(it2delta2) {$\delta$};
   \draw[black,latex-] (it2k.east |- it2j) -- (it2delta2.west);
   \draw[black] (it2delta2.east) -- (6.7cm + 13cm*0.25, 5cm |- it2j);

   \node[font=\footnotesize,color=black,anchor=west] at (0.0cm,-6.1cm) {$\dots$};

\end{tikzpicture}
\end{center}
\caption{A first few possible iterations of the co-ranking
  algorithm. In the first iteration only the first condition can be
  true, so the interval from $j_{\rm low}$ to $j$ is halved. The lower bound
  $k_{\rm low}$ is set accordingly, and $k$ is increased by the same $\delta$. In
  the second iteration, assuming the second condition to be true,
  instead the interval from $k_{\rm low}$ to $k$ is halved. 
}
\label{fig:searching}
\end{figure}
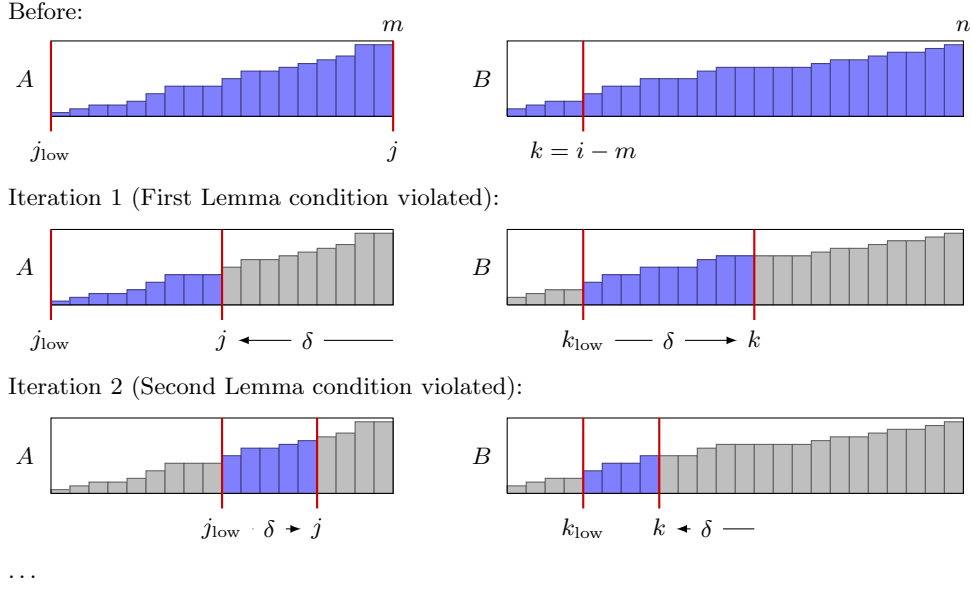

\begin{proposition}
Algorithm~\ref{alg:corank} computes the co-ranks $j$ and $k$
for ordered arrays $A$ and $B$ with $m$ and $n$ elements for any
$0\leq i\leq m+n$. The algorithm requires at most
$\ceiling{\log_2\min(m,n,i,m+n-i)}$ iterations and thus comparisons
of element keys.
\end{proposition}

\begin{proof}
The algorithm clearly maintains the invariant $j+k=i$. We claim
further that either the co-rank for array $A$ lies between $j_{\rm
  low}$ and $j$ or the co-rank for array $B$ lies between $k_{\rm
  low}$ and $k$. This holds before the first iteration, as $j_{\rm
  low}$ and $j$ are set assuming that as many of the $i$ output
elements as possible come from $A$ (see Figure~\ref{fig:searching}).
Assume that the invariant holds before an iteration starts.  If the
condition in Line 6 holds, then the correct index in array $A$ must be
between $j_{\rm low}$ and $j$. For the next iteration the correct
index in $A$ is either between $j_{\rm low}$ and $j-\delta$ or between
$j-\delta$ and $j$. In the latter case, the co-rank for array $B$ must
be between $k$ and $k+\delta$, and the invariant is maintained by
setting the lower bound $k_{\rm low}$ to $k$ and increasing $k$ by
$\delta$. Similarly, if instead the condition in line $11$ evaluates
to true, the correct index in $B$ must be between $k_{\rm low}$ and
$k$, and dividing the interval and increasing the lower bound $j_{\rm
  low}$ for array $A$ likewise reestablishes the invariant. It
therefore holds also after the iteration. After the first iteration,
$k-k_{\rm low}=\delta$ (if the first condition is true, otherwise the
algorithm terminates) and $j-j_{\rm low}\leq\delta$. As $j$ is
initialized to $\min(i,m)$ and $j_{\rm low}\geq 0$, trivially
$j-j_{\rm low}\leq\min(i,m)$. In order for $k$ not to exceed $n$,
$j_{\rm low}$ must be initialized such that $k+(j-j_{\rm low})<n$,
that is $i-j_{\rm low}<n$; $j_{\rm low}$ is therefore set to
$\max(0,i-n)$. Assume $j_{\rm low}=i-n\geq 0$; since $j$ is
initialized to either $m$ or $i<m$, it follows that either $j-j_{\rm
  low}\leq m+n-i$ or $j-j_{\rm low}\leq n$.  As $\delta$ is halved in
each subsequent iteration and $j-j_{\rm low}$ starts out being
$\min(m,n,i,m+n-i)$, at most $\ceiling{\log_2 \min(m,n,i,m+n-i)}$
iterations are required.
\end{proof}

\section{Parallel merging}

\begin{algorithm}
\caption{Synchronization-free parallel merging of ordered arrays $A$ and $B$ 
for processing element $r$, $0\leq r<p$}\label{alg:merge}
\begin{algorithmic}[1]
\STATE $i_r\gets \floor{r\frac{m+n}{p}}$ \COMMENT{Start index of
  output block}
\STATE $i_{r+1}\gets \floor{(r+1)\frac{m+n}{p}}$ \COMMENT{End
  index of output block}

\COMMENT{To avoid synchronization processing element $r$ computes
  co-ranks
for both start and end index}
\STATE $(j_r,k_r) \gets \corank(i_r,A,m,B,n,\leq)$
\STATE $(j_{r+1},k_{r+1}) \gets \corank(i_{r+1},A,m,B,n,\leq)$
\STATE $\merge(A[j_r,\ldots,j_{r+1}-1],B[k_r,\ldots,k_{r+1}-1],C[i_r,\ldots,i_{r+1}-1],\leq)$
\end{algorithmic}
\end{algorithm}

The co-ranking algorithm provides a simple and efficient way of
performing merging in parallel. Let $p$ processing elements be given,
all of which can access input and output arrays $A$, $B$ and $C$. Each
processing elements has an own id $r, 0\leq r<p$. Each processing
element independently computes the start and end indices of a block of
the output array from $C[i_r,\ldots i_{r+1}-1]$. The output blocks can be
chosen such that they partition the whole output array, and differ in
size by at most one element. Each processing element computes for both
start and end index the corresponding co-ranks. These co-ranks determine
the (disjoint) blocks of the input arrays this processor need to merge
sequentially to compute its output block.  This is shown in detail as
Algorithm~\ref{alg:merge}.

\begin{proposition}
Algorithm~\ref{alg:merge} merges ordered arrays $A$ and
$B$ of $m$ and $n$ elements stably using $p$ processing elements.
The number of elements to merge per processing element is at most 
$\ceiling{\frac{m+n}{p}}$, and the total time complexity is
$\mathcal{O}(\frac{n+m}{p}+\log\min(m,n))$.
\end{proposition}

Algorithm~\ref{alg:merge} avoids synchronization by letting each
processor compute the co-ranks for both start and end index. If
synchronization is inexpensive, half the co-ranking work can be saved
by letting each processing element read the co-ranks for its end index
from the next processing element. A synchronization step is in this
case required after the co-ranks computation. It is also worth noting
that in the co-ranking procedure, when invoked in parallel by several
processing elements, concurrent reading of locations may easily
occur. However, it may well be possible to eliminate these by a careful
pipelining such as in \cite{Chen95:pbs,HagerupRub89}. The
implementation given here will run efficiently on a CREW PRAM.

Stability follows from the properties of the co-ranking procedure, and
the use of a stable, sequential merge algorithm. The number of
elements to merge per processing element differs at most by one
element. These are the main improvements over previous parallel merge
algorithms, where the number of elements to merge, although
$\mathcal{O}((n+m)/p)$, can differ by a factor of two.

Algorithm~\ref{alg:merge} assumes a shared-memory parallel system.
The algorithm can, however, easily be adapted to distributed memory
systems as shown in~\cite{Traff12:merge}. This is a considerably more
elegant, better, and easier implementation than for instance the BSP
algorithm presented in~\cite{GerbessiotisSiniolakis01}.

\end{document}